\documentclass[12pt]{article}
\usepackage{graphicx}
\usepackage{amssymb}
\usepackage{amsmath}
\newcommand{\norm}[1]{\left\Vert#1\right\Vert}

\newcommand{\ve}{\varepsilon}

\numberwithin{equation}{section}
\def\p{{\partial}}

\def\BR{{\mathbb R}}

\def\clc{{\mathcal C}}

\def\nn{\nonumber}
\newtheorem{Pa}{Paper}[section]
\newtheorem{Tm}[Pa]{{\bf Theorem}}

\newtheorem{Cy}[Pa]{{\bf Corollary}}
\newtheorem{Rk}[Pa]{{\bf Remark}}

\newtheorem{Dn}[Pa]{{\bf Definition}}
\newtheorem{Pn}[Pa]{{\bf Proposition}}
\usepackage{amsmath}

\title{ Feynman's  linear divergence problem}

\author{Alexander Sakhnovich\footnote{Faculty of Mathematics,
University
of
Vienna, Oskar-Morgenstern-Platz 1, A-1090 Vienna,
Austria;
E-mail: oleksandr.sakhnovych@univie.ac.at} \, and \, Lev Sakhnovich\footnote{Retired from Courant institute; 99 Cove ave. Milford, CT 06461, USA; E-mail: lsakhnovich@gmail.com} }

\date{}

\begin{document}

\maketitle




\begin{abstract} First, we consider  generalized wave and scattering operators and derive modifications  of   commutation relations (between scattering operators and unperturbed operators) when the corresponding deviation factors behave as  $\exp\{i t \clc_{\pm}\}$
for $t\to \pm \infty$.
Then, we construct so called  secondary generalized scattering operators for the related case of linear divergence in QED, which
gives a positive answer  (in that case) to the well-known problem of  J. R. Oppenheimer regarding scattering operators in QED:
``Can the procedure be freed of the expansion in $\varepsilon$ and carried out rigorously?"
\end{abstract}

\section{Introduction} \label{intro}
{\bf 1.} The problem of the divergences in the scattering theory for quantum electrodynamics has a long history starting in the 1940s with  the fundamental works 
by
R. Feynman \cite{Fey0, Fey1} (see also  \cite{AB} and the references therein). The divergences belong  to one of the  three types:
to the logarithmic,  linear or  quadratic cases (see \cite{AB}). The  logarithmic case was rigorously treated  
in the paper \cite {Sakh4} by one of the coauthors. This work may be considered as a continuation of \cite{Sakh4}, and the linear divergencies are studied here.

The presentation in \cite{Sakh4}  is started with  the introduction of the self-adjoint operators $A,\,A_0,\,A_1$ which are defined in Hilbert space $H$ and are connected by the relation
\begin{equation} A=A_0+{\varepsilon}A_1. \label{1.1}\end{equation}
Here, $A$ is a perturbed operator, $A_0$ is an unperturbed operator, and $A_1$ is a perturbation operator. Next, we introduce the operator function
\begin{equation} S(t,\tau,\varepsilon)=\exp(i t A_0)\exp(-i tA)\exp(i\tau A))\exp(-i\tau A_0), \label{1.2}\end{equation}
which is closely related to the scattering operator $S(A,A_0)$ (see \cite{Sakh2, Sakh1, Sakh4}).
 It follows from \eqref{1.2} that
\begin{equation}\frac{\partial}{\partial{t}}S(t,\tau,\varepsilon)=-i{\varepsilon}V(t)S(t,\tau,\varepsilon),\quad
\frac{\partial}{\partial{\tau}}S(t,\tau,\varepsilon)=i{\varepsilon}S(t,\tau,\varepsilon)V(\tau), \label{1.3}\end{equation}
where
\begin{equation} V(t)=\exp(itA_0)A_1\exp(-itA_0),\quad S(t,t,\varepsilon)=I,\label{1.4}\end{equation}
and $I$ is the identity operator.
According to \eqref{1.4},  the  self-adjoint operator function $V(t)$  may be considered (at each $t$) as a  special representation of the perturbation operator $A_1$.
\begin{Dn}\label{Definition 1.2}The self-adjoint operator function $V(t)$ in \eqref{1.3} is  called the perturbation operator function.\end{Dn} 

Relations \eqref{1.3} and \eqref{1.4} are used in Section \ref{ApA}  in order to derive commutative properties \eqref{4.8} and \eqref{4.9} (for the {\it generalized scattering operators} $S(A,A_0)$) as well
as  the generalizations \eqref{4.7} and \eqref{4.9+} of \eqref{4.8} and \eqref{4.9},  respectively.  Recall that the commutation with the unperturbed operator $A_0$ is one of the main properties of the scattering operator.
Equalities \eqref{4.7}--\eqref{4.9+} are modified commutative relations for the considered here case of  the  generalized scattering operators.

Further, in Sections \ref{sec2} and \ref{secEx}  we do not assume that the operators  $A_0$ and $A_1$ are known or
exist and directly use  the equation \eqref{1.3}.  In this way, we, similar to \cite{Sakh4}, follow the suggestions of Heisenberg’s S-program  \cite{Hei}.
In Section \ref{sec2}, we consider  perturbation operator function $V(t)$ of the form \eqref{3.6}--\eqref{3.6+}. For instance, we have
$$V(t)=C_{+}+\frac{B_+}{t}+u(t) \quad {\mathrm{for}} \quad t \geq 1,$$
whereas we had logarithmic divergence with $C_+=0$ in \cite{Sakh4}.
In Section~\ref{secEx},
the perturbation operator function is given by the formula $V(L,q)=C_+ + \frac{1}{L}B_+(q)+u(L,q)$, where $t$ is  substituted by the length $L$ $(L\geq 1)$
(space-time approach \cite{Fey1}) in the formula for $V(t)$ above and the corresponding operator $B_+$ is the operator of multiplication by $B_+(q)$ $(q\in \BR^4)$.   Here, the notation $\BR$ stands for the real axis. 

In Theorem \ref{Theorem 3.5} we construct  the so called  secondary  generalized    scattering operator  $S^R(+\infty,-\infty,\ve)$.
In formulas \eqref{u16} and \eqref{u17},  the   secondary  generalized    scattering operator  $S^R(+\infty,\ve)$ is constructed for the ultraviolet case.
In both cases, we obtain a positive answer to the well-known problem of
J. R. Oppenheimer \cite{Opp} regarding scattering operators in QED: ``Can the
procedure be freed of the expansion in $\ve$ and carried out rigorously?”

{\bf Notations.} 
The notation   $A^*$ stands for the  adjoint to $A$ operator and the symbol $\Longrightarrow$ means the operator convergence by norm.
\section{Generalized wave operators,  deviation \\ factors and commutation property}\label{ApA}
{\bf 1.} Wave operators play an essential role in many problems of mathematical physics.
However, the wave operators do not exist when the initial and/or final states of the system may not be regarded as free.
In these cases, one has to consider the {\it generalized wave operators}  (see, e.g.,  \cite{Do, Sakh2,  Sakh3, Sakh10, Sakh10+, Sak}
as well as Appendix A in \cite{Sakh4}).

Let $A$ and $A_0$ be linear self-adjoint (not necessary bounded) operators with  absolutely continuous spectrum acting in some Hilbert space $H$. 
Then, the   generalized wave operators $W_{\pm}(A,A_{0})$ and  the unitary deviation factor (operator function) $W_{0}(t)$ acting in $H$ are defined
in the following way.
\begin{Dn}\label{Dn7.1}
 Operators  $W_{\pm}(A,A_{0})$ and operator function  $W_{0}(t)$ are called generalized wave operators and a deviation factor, respectively, if
the following conditions are fulfilled:

1. The limits
\begin{equation}W_{\pm}(A,A_{0})
=s{-}{\lim_{t{\to}\pm\infty}}\, [e^{iAt}e^{-iA_{0}t}W_{0}^{-1}(t)]
\label{7.1}\end{equation}
exist in the sense of strong convergence $(s{-}{\lim})$.

2. The operators $W_{0}(t)$ and  $W_{0}^{-1}(t)$ 
are unitary for all $t\in \BR$ and
\begin{equation}\lim_{t{\to}\pm\infty}W_{0}(t+\tau)W_{0}^{-1}(t)=\exp\{i\tau \clc_{\pm}\} \quad (\clc_{\pm}=\clc_{\pm}^*) \quad {\mathrm{for\, all}}\quad \tau=\overline{\tau}.
\label{7.2}\end{equation}

3. The following commutation relations  hold:
\begin{align}& W_{0}(t)A_{0}=A_{0}W_{0}(t)\,\, {\mathrm{for}} \,\, t\in \BR; \label{7.3}
\\ & 
W_{0}(t)W_{0}(t+\tau)=W_{0}(t+\tau)W_{0}(t) \,\, {\mathrm{for}}\,\, t,\tau \geq 0 \,\, {\mathrm{and \,\, for}}\,\, t,\tau \leq 0.
\label{7.3+}\end{align}
\end{Dn}
If  $W_{0}(t)\equiv I$, then the operators $W_{\pm}(A,A_{0})$ are usual wave
operators.
\begin{Rk}\label{RkDn} Definition \ref{Dn7.1}, which will be useful in the considerations related to linear divergence,
is more general than in the previous works $($compare, e.g., with \cite{Sakh2} and \cite[Appendix A]{Sakh4}$)$.
\end{Rk}
\begin{Rk}\label{RkC} Relations \eqref{7.2} and \eqref{7.3} imply the commutativity
\begin{equation}A_0\exp\{i\tau \clc_{\pm}\} =\exp\{i\tau \clc_{\pm}\}A_0 \quad  {\mathrm{for\, all}}\quad \tau=\overline{\tau}.
\label{7.2+}\end{equation}
\end{Rk}
\begin{Pn}\label{Pn4.2} The generalized wave operators have the property
 \begin{equation}\exp\{i \tau A\}W_{\pm}(A,A_0)=W_{\pm}(A,A_0)\exp\{i \tau (A_{0}+\clc_{\pm})\}.\label{4.4} \end{equation}\end{Pn}
 {\it Proof.}
 It follows from \eqref{7.1} and \eqref{7.3} that 
\begin{align} \nn \exp\{i\tau A\}W_{\pm}(A,A_0)= & s{-}{\lim_{t{\to}\pm\infty}} \exp \{ i (t+\tau)A\}\exp\{-i(t+\tau)A_{0}\}W_{0}^{-1}(t)
\\ & \nn \times \exp\{i \tau A_{0}\}
\\ 
= & \nn s{-}{\lim_{t{\to}\pm\infty}} \exp \{ i (t+\tau)A\}\exp\{-i(t+\tau)A_{0}\}W_{0}^{-1}(t+\tau)
\\  & \times  \big(W_0(t+\tau)W_0^{-1}(t)\big)\exp\{i \tau A_{0}\}.
\label{4.2} \end{align}
Hence, \eqref{7.1} and \eqref{7.2}, \eqref{7.2+} yield \eqref{4.4}. $\Box$

According to Definition \ref{Dn7.1}, the generalized wave operators are unitary. Therefore, our next assertion follows from  \eqref{4.4}.
\begin{Cy}\label{Corollary 4.3}The operators $\exp\{i \tau A\}$ and $\exp\{i \tau (A_{0}+\clc_{\pm})\}$ are unitarily equivalent.\end{Cy}
\begin{Dn}\label{Definition 7.2}The generalized scattering operator $S(A,A_0)$
is defined by the formula:
\begin{equation}S(A,A_0)=W^{*}_{+}(A,A_0)W_{-}(A,A_0),\label{7.4}
\end{equation}
where $W_{\pm}(A,A_{0})$ are generalized wave operators.\end{Dn}
Relations \eqref{4.4} and \eqref{7.4} yield the next corollary.
\begin{Cy}\label{Cy4.6} The following equality is valid for the generalized scattering operator $S(A,A_0):$
\begin{equation}S(A,A_0)=\exp\{-i \tau(A_0+\clc_+)\}S(A,A_0)\exp\{i\tau(A_0+\clc_-)\}. \label{4.7}\end{equation}\end{Cy}
If $\clc_+=\clc_-=\clc$, formula \eqref{4.7} turns into a commutative formula:
\begin{equation}\exp\{i \tau(A_0+\clc)\}S(A,A_0)=S(A,A_0)\exp\{i\tau(A_0+\clc)\}. \label{4.8}\end{equation}

\begin{Rk}\label{RkCom} If $\clc_+=\clc_-=\clc$,  commutative formula \eqref{4.8} yields the commutative property
\begin{equation}(A_0+\clc)S(A,A_0)=S(A,A_0)(A_0+\clc). \label{4.9}\end{equation}
If the equality  $\clc_+=\clc_-$ is not required, formula \eqref{4.7} implies a more general relation:
\begin{equation}(A_0+\clc_+)S(A,A_0)=S(A,A_0)(A_0+\clc_-). \label{4.9+}\end{equation}
\end{Rk}
\section{Secondary  generalized    scattering  \\  and perturbation operators, \\ linear divergence}\label{sec2}
{\bf 1.} Recall that the operator function $S(t,\tau,\varepsilon)$ given by \eqref{1.2} satisfies \eqref{1.3}:
\begin{equation}\frac{\partial}{\partial{t}}S(t,\tau,\varepsilon)=-i{\varepsilon}V(t)S(t,\tau,\varepsilon),\quad
\frac{\partial}{\partial{\tau}}S(t,\tau,\varepsilon)=i{\varepsilon}S(t,\tau,\varepsilon)V(\tau), \label{3.2}\end{equation}
where $t$ and $\tau$ belong to the real axis  $\BR$ and $V(t)$ is given by \eqref{1.4}.

Further in the text, the operators $A,\, A_1,\,A_0$, which appear in \eqref{1.1}, \eqref{1.2} and \eqref{1.4}, are unknown
(or even do not exist). Thus, we do not suppose that the operator $V(t)$ in \eqref{3.2} has the form \eqref{1.4}. However, similar to \cite{Sakh4} we assume that \eqref{3.2} holds for some
self-adjoint perturbation operator function $V(t)$ and that
\begin{equation} S(t,t,\varepsilon)=I.\label{3.3}\end{equation}
 The operator $V(t)$ itself is now called a perturbation operator. Using the method of successive approximation and relations \eqref{3.2}, \eqref{3.3},
we (similar to \cite[Proposition 3.1]{Sakh4}) obtain our next proposition.
\begin{Pn}\label{Proposition 3.1} Assume that  $V(t)$ is a self-adjoint, continuous and bounded operator function in the domain $-T\leq t \leq T$.  Then, there exists such 
$\varepsilon _{T}>0$ that the series
\begin{equation} S(t,\tau,\varepsilon)=\sum_{p=0}^{\infty}S_{p}(t,\tau)\varepsilon^{p}\qquad (S_0(t,\tau) \equiv I) \quad \label{3.4}\end{equation}
is convergent in the domain  $-T \leq \tau\leq  t\leq T$  for  $|\varepsilon| \leq \varepsilon _{T}$. \end{Pn}
Thus, taking into account \eqref{3.2} and the equality $S_{p+1}(t,t)=0$ for $p\geq 0$, we have
\begin{align} & S_{p+1}(t,\tau)=-i\int_{\tau}^{t}V(t_1)S_p(t_1,\tau)dt_1=-i\int_{\tau}^{t}S_p(t,\tau_1)V(\tau_1)d\tau_1.\label{3.5}\end{align}
Hence, for the first approximation $S_1$ we derive
\begin{align} & S_1(t,\tau)=-i\int_{\tau}^{t}V(t_1)dt_1.\label{3.5+}\end{align}

In the present article, we consider the operator function $V(t)$  of the form
\begin{align} \label{3.6} & V(t)=C_{+}+\frac{B_+}{t}+u(t) \quad {\mathrm{for}} \quad t \geq 1,
\\ 
\label{3.7} & V(t)=u(t) \quad {\mathrm{for}} \quad |t| \leq 1,
\\  &
V(t)=C_{-}+\frac{B_-}{t}+u(t)  \quad {\mathrm{for}} \quad t \leq -1, \label{3.6+}\end{align}
where $B_{\pm}$ and $C_{\pm}$ are self-adjoint, bounded operators and $u(t)$ is a self-adjoint, continuous operator function such that
\begin{equation} \norm{u(t)}=O(|t|^{-\nu}) \quad {\mathrm{for}} \quad  |t| \to  \infty \qquad (\nu>1).\label{3.8}\end{equation}
In particular, relations \eqref{3.5+}--\eqref{3.6+} imply that
\begin{align}& S_1(t,\tau)=-i\left(C_{+}(t-\tau)+B_{+}\ln(t/\tau)+\int_{\tau}^{t}u(s)ds\right)\quad {\mathrm{for}} \quad t \geq \tau \geq 1,  \label{3.9}
\\ &
S_1(t,\tau)=i\left(C_{-}(\tau-t)+B_{-}\ln(\tau/t)-\int_{\tau}^{t}u(s)ds\right)\quad {\mathrm{for}} \quad \tau \leq t \leq -1. \label{3.10}\end{align}
Let us introduce  (see \cite{Sak, Sakh4} on this topic) the continuous deviation factors  given by
\begin{align}\label{3.11} &\frac{d}{dt}W_{0}(t,\varepsilon)=i{\varepsilon}W_{0}(t,\varepsilon)V_{0}(t)\quad (t\in \BR),\quad  W_0(0)=I,
\end{align}
where
\begin{align}&   V_0(t)=C_{+}+\frac{B_+}{t} \quad (t \geq 1), \quad
V_0(t)=C_{-}+\frac{B_-}{t}\quad (t\leq -1) \label{3.12}, 
\\ &
V_0(t)=0 \quad (|t|<1). \label{3.12+}
\end{align}
The next  assertion  follows from \eqref{3.11}--\eqref{3.12+}.
\begin{Pn}\label{Proposition 3.2} If $B_{\pm}C_{\pm}=C_{\pm}B_{\pm}$, then 
\begin{align}& \label{z1} W_0(t,\varepsilon)=\exp\{i\varepsilon\big(C_{+}(t-1)+B_{+}\ln(t)\big)\} \quad {\mathrm{for}} \quad t \geq 1,
\\ & \label{z2}
W_0(t,\varepsilon)=\exp\{i\varepsilon\big(C_{-}(t+1)+B_{-}\ln(|t|)\big)\} \quad {\mathrm{for}} \quad  t \leq -1, 
\\ & \label{z3}
W_0(t,\varepsilon)\equiv I \quad {\mathrm{for}} \quad  |t| \leq 1.
\end{align}
\end{Pn}
Let us introduce the regularization $S^{R}(t,\tau,\varepsilon)$ of $S(t,\tau,\varepsilon)$:
\begin{equation}S^{R}(t,\tau,\varepsilon)=W_{0}(t,\varepsilon)S(t,\tau,\varepsilon)W_{0}^{-1}(\tau,\varepsilon) \label{3.16}\end{equation}
It follows from the formulas \eqref{3.2},  \eqref{3.3}, \eqref{3.6}--\eqref{3.6+}, and \eqref{3.11}--\eqref{3.12+}  that
\begin{align}& \frac{\partial}{\partial{t}}S^{R}(t,\tau,\varepsilon)=-i{\varepsilon}U(t,\varepsilon)S^{R}(t,\tau,\varepsilon),\quad
S^{R}(t,t,\varepsilon)=I,\label{3.17}
\\ &
\frac{\partial}{\partial{\tau}}S^{R}(t,\tau,\varepsilon)=i{\varepsilon}S^{R}(t,\tau,\varepsilon)U(\tau,\varepsilon), \quad S^{R}(\tau,\tau,\varepsilon)=I,\label{3.18}\end{align}
where
\begin{equation}U(t,\varepsilon)=W_{0}(t,\varepsilon)u(t)W_{0}^{-1}(t,\varepsilon)\quad (t\in \BR) .\label{3.19}\end{equation}
According to \eqref{z3} and \eqref{3.16}, we have
\begin{align}&\nn
S^{R}(t,\tau,\varepsilon)=S(t,\tau,\varepsilon) \quad {\mathrm{for}} \quad |t|\leq 1,\,\, |\tau|\leq 1.
\end{align}
Using relations \eqref{3.17} and \eqref{3.18}  as well as multiplicative integrals (see, e.g., \cite{DF} or \cite{Sakh4} and references on the topic therein)  we 
write down an important  representation of  $S^{R}(t,\tau,\varepsilon)$.
\begin{Tm}\label{Theorem 3.3}
Let the regularized operator function $S^{R}(t,\tau,\varepsilon)$ be given by \eqref{3.16}.
Then, we have a multiplicative integral reperesentation
\begin{equation} S^{R}(t,\tau,\varepsilon)= \overset{t}{\overset{\curvearrowleft}{\underset{\tau}{\int}}}
e^{-i\varepsilon{U(s,\varepsilon)ds}} \qquad (t\geq \tau).
\label{3.21}\end{equation}
\end{Tm}
In order to study the asymptotics of  $S^{R}(t,\tau,\varepsilon)$  one may use the inequality
\begin{equation} \Big\| \overset{t}{\overset{\curvearrowleft}{\underset{\tau}{\int}}}e^{F(t)dt}\Big\| \leq \exp\left(\int_{\tau}^{t}\norm{F(t)}dt\right), \quad
{\mathrm{where}} \quad  \tau<t,\label{8.5}\end{equation}
as well as Proposition B.1 and Corollary B.2 in \cite{Sakh4}, which follow from \eqref{8.5}.

It is easy to see that $W_0(t)$ given by \eqref{3.11} is unitary. Hence, $U(t,\ve)$ of the form \eqref{3.19}
is self-adjoint. Therefore, Theorem \ref{Theorem 3.3} implies the following proposition.
\begin{Pn}\label{Proposition 3.4} The operators $S^{R}(t,\tau,\varepsilon)$ $(t\geq \tau)$
are unitary.\end{Pn}

Recall that $\Longrightarrow$ means the convergence by norm.
In view of \eqref{3.8} and  \eqref{3.19}--\eqref{8.5}, the main theorem of this section  is valid (see below).
\begin{Tm}\label{Theorem 3.5} Let  $S^{R}(t,\tau,\varepsilon)$ be given by \eqref{3.16},
and assume that relations \eqref{3.2}, \eqref{3.3}, \eqref{3.6}--\eqref{3.8} and \eqref{3.11}--\eqref{3.12+} hold.
  Then, we have  
  \begin{equation}S^{R}(t,\tau,\varepsilon){\Longrightarrow}S^{R}(+\infty,-\infty,\varepsilon),\quad t{\to}+\infty,\quad
 \tau{\to}-\infty.\label{3.26}
\end{equation}
\end{Tm}
Similar to the case in \cite{Sakh4}, 
the unitary operator  $S^{R}(+\infty,-\infty,\varepsilon)$ is called the secondary generalized scattering operator.
\begin{Rk}\label{Remark 3.8} In classical quantum theory, the  operator function $V(t)$ has the form \eqref{1.4}. In our theory, the operators 
$A,\, A_1, \, A_0$ are unknown and we don't suppose that operator $V(t)$ has the form  \eqref{1.4}. 
In our approach, the perturbation operator $V(t)$ is defined using the first approximation $S_1(t,\tau)$.
Indeed, formula \eqref{3.5+} yields
\begin{equation}V(t)=i\frac{\p}{\p t}S_1(t,\tau).
\label{3.30}\end{equation}
\end{Rk}
\section{Ultraviolet example} \label{secEx}
{\bf 1.} In various important ultraviolet examples (see, e.g.,  \cite[Sections 4,5]{Sakh4} and references
therein) the variable $t$ is substituted by  the length variable  $L$  and $\tau$  is substituted by $q=\begin{bmatrix}q_1 & q_2 & q_3 & q_4\end{bmatrix}\in \BR^4$. Instead of
the operators $S(t,\tau,\ve)$ and  $S_1(t,\tau)$, we now consider the operators $S(L, \ve)$ and its first approximation $S_1(L)$, which are defined as the  multiplications by the functions $S(L,q,\ve)$ and $a_1(L,q)$,
respectively.
Moreover, $a_{1}(L,q)$ may be written down as an integral over the four dimensional sphere with radius $L$ (in spherical coordinates):
\begin{equation}
a_{1}(L,q)=-i\int_{1}^{L}\int_{0}^{2\pi}\int_{0}^{\pi}\int_{0}^{\pi}F(p,q)r^{3}\big(\sin^{2}\phi_{1}\big)\sin\phi_2d\phi_1d\phi_2d\phi_3dr,\label{u2}
\end{equation}
where $p=\begin{bmatrix}p_1 & p_2 & p_3 & p_4\end{bmatrix}\in \BR^4$ and
\begin{align}& p_1=r\cos\phi_1,\quad p_2=r\sin\phi_1\cos\phi_2, \label{4.3-}
\\ &
p_3=r\sin\phi_1\sin\phi_2\cos\phi_3, \quad
p_4=r\sin\phi_1\sin\phi_2\sin\phi_3 .\label{4.3}\end{align}
Interesting results on the asymptotic behaviour  of the Feynman integrals \eqref{u2} are derived in \cite{Du} (see also some other papers and references
in \cite{AC}).
The considerations of Section \ref{sec2} are easily reformulated for this case (space-time approach),
see also \cite[Section 4]{Sakh4}. In particular, we set
\begin{align}& \label{u4}
\frac{\partial}{\partial{L}}S(L,q,\varepsilon)=-i{\varepsilon}V(L,q)S(L,q,\varepsilon) \quad (L\geq 1),\quad S(1,q,\ve)=1, 
\\ & \label{4.4+}
 V(L,q)=i\frac{\p}{\p L}a_1(L,q).
\end{align}
It follows from \eqref{u2}--\eqref{4.3} and \eqref{4.4+} that
\begin{align} 
& \label{4.5}
 V(L,q)=\int_{0}^{2\pi}\int_{0}^{\pi}\int_{0}^{\pi}F(P,q)L^{3}\big(\sin^{2}\phi_{1}\big)\sin\phi_2d\phi_1d\phi_2d\phi_3,\\
 & P=\begin{bmatrix}P_1 & P_2 & P_3 & P_4\end{bmatrix}\in \BR^4, \quad
 P_1=L\cos\phi_1,\quad P_2=L\sin\phi_1\cos\phi_2, \nn
\\ &
P_3=L\sin\phi_1\sin\phi_2\cos\phi_3, \quad
P_4=L\sin\phi_1\sin\phi_2\sin\phi_3 .\nn
\end{align}
In the case of 
$$F(p,q)=\frac{p_{\mu}}{\big(p^2-2pq+\ell(q)\big)^2},$$
where $p^2$ and $pq$
are scalar products, we have the so called superficial linear divergence (that is, we have logarithmic
divergence, see \cite[Example 5.3]{Sakh4}).

{\bf 2.} Further, we will consider the case of linear divergence, where
\begin{align}& \label{4.6}
F(p,q)=\frac{|p_{\mu}|}{\big(p^2-2pq+\ell(q)\big)^2}, \quad \ell(q)>q^2,
\end{align}
and   $\ell(q)$ is a continuous function.   It follows from \eqref{4.3} that 
$$p_3^2+p_4^2=   r^2\sin^2\phi_1\sin^2\phi_2.$$ 
Hence,
taking into account \eqref{4.3-} we obtain $p^2=r^2$.      Therefore, we have
\begin{align}& \label{u7}
\frac{r^4}{\big(p^2-2pq+\ell(q)\big)^2}=\left(1-\frac{2}{r^2}pq+\frac{\ell(q)}{r^2}\right)^{-2}.
\end{align}
We also have
\begin{align}\nn
\left(1-\frac{2}{r^2}pq+\frac{\ell(q)}{r^2}\right)^{2} &=1-2\left(\frac{2}{r^2}pq-\frac{\ell(q)}{r^2}\right)+\left(\frac{2}{r^2}pq-\frac{\ell(q)}{r^2}\right)^2
\\ &  \label{u8}
=1-\frac{4}{r^2}pq+O\left(\frac{1}{r^2}\right) \quad {\mathrm{for}} \quad r\to \infty .
\end{align}
Relations \eqref{u7} and \eqref{u8} yield
\begin{align}& \label{u9}
\frac{r^4}{\big(p^2-2pq+\ell(q)\big)^2}=1+\frac{4}{r^2}pq+O\left(\frac{1}{r^2}\right).
\end{align}

We will set $\mu=1$ in \eqref{4.6} (with $p_1$ given by \eqref{4.3-}) since other cases are all treated in a similar way. Then,
formulas  \eqref{4.5}, \eqref{4.6} and \eqref{u9} imply that
\begin{align}& \label{4.10}
a_1(L,q)=-i \left(C_+ (L-1) +B_+(q)\ln L+\int_1^L u(r,q)dr \right),
\end{align}
where
\begin{align}& \label{4.11}
 C_+=\int_{0}^{2\pi}\int_{0}^{\pi}\int_{0}^{\pi}|\cos \phi_1|\big(\sin^{2}\phi_{1}\big)\sin\phi_2 d\phi_1d\phi_2d\phi_3,
\\  & \label{4.12}
 B_+(q)=4\int_{0}^{2\pi}\int_{0}^{\pi}\int_{0}^{\pi}(pq/r)|\cos \phi_1|\big(\sin^{2}\phi_{1}\big)\sin\phi_2d\phi_1d\phi_2d\phi_3,
 \\ &  \label{4.14} \|u(r, q)\| = O\big(r^{-2}\big).
 \end{align}
Here, the operator $B_+$ is the operator  of multiplication by $B_+(q)$ in the space of
the functions of $q$, where $|q| \leq M<\infty$, $C_+$ is the multiplication by constant in the same space, and $u(r)$ is  the multiplication
by $u(r,q)$ (compare \eqref{4.10} with \eqref{3.9}). Note also that $B_+(q)$ given by \eqref{4.12} does not depend on $r$ because the expression $p/r$ does not depend on $r$.
In view of \eqref{4.12}, the operator $B_+$ is bounded. Clearly, $B_+$ and $C_+$  commute. Finally, \eqref{4.4+} and \eqref{4.10} yield
\begin{align}& \label{u13}
V(L,q)=C_+ + \frac{1}{L}B_+(q)+u(L,q).
\end{align}

In analogy with \eqref{z1},  we introduce the
deviation factor by the formula
 \begin{align}    & \label{4.15}
W_0(L,q,\ve)=\exp\big\{i\ve\big(C_+(L-1)+B_+(q)\ln L\big)\big\} \quad {\mathrm{for}} \quad L\geq 1.
\end{align}
Similar to the case of ultraviolet divergence in \cite{Sakh4}, a regularized scattering function $S^R$ is introduced  by the formula
\begin{equation}S^{R}(L,q,\varepsilon)=W_{0}(L,q,\varepsilon)S(L,q,\varepsilon), \label{u10}\end{equation}
and the  regularized scattering operator  $S^{R}(L,\varepsilon)$ is the operator of multiplication by $S^{R}(L,q,\varepsilon)$.
It follows from \eqref{u4} and \eqref{u13}--\eqref{u10} that
\begin{equation}\frac{\p}{\p L}S^{R}(L,q,\varepsilon)=-i{\varepsilon}U(L,q,\varepsilon)S^{R}(L,q,\varepsilon),\quad
S^{R}(1,q,\varepsilon)=1,\label{u11}\end{equation}
where
\begin{equation}U(L,q,\varepsilon)=W_{0}(L,q,\varepsilon)u(L,q)W_{0}^{-1}(L,q,\varepsilon),\quad L \geq 1.\label{u12}\end{equation}
Using multiplicative integrals (see \cite[Appendix B]{Sakh4} and the references therein) we represent $S^{R}(L,q,\varepsilon)$ in the
form:
\begin{equation} S^{R}(L,q,\varepsilon)= \overset{L}{\overset{\curvearrowleft}{\underset{1}{\int}}}
e^{-i\varepsilon{U(r,q,\varepsilon)dr}},\quad L \geq 1.\label{u15}\end{equation}

{\it In view of  \eqref{4.14} and \eqref{u11}--\eqref{u15} we obtain our result for the secondary generalized scattering operator $S^{R}(+\infty,\varepsilon)$
of multiplication by $S^{R}(+\infty,q,\varepsilon)$}:
\begin{equation}S^{R}(L,\varepsilon){\Longrightarrow}S^{R}(+\infty,\varepsilon)\quad 
{\mathrm{for}} \quad L{\to}+\infty,\label{u16}
\end{equation}
{\it where the symbol $\Longrightarrow$ denotes convergence by norm and}
\begin{equation}S^{R}(+\infty,q,\varepsilon)= \overset{+\infty}{\overset{\curvearrowleft}{\underset{1}{\int}}}
e^{-i\varepsilon{U(r,q,\varepsilon)dr}}. \label{u17}\end{equation}

\end{document}